\pdfoutput=1

\documentclass[runningheads]{llncs}

\usepackage{graphicx}
\usepackage{epstopdf}
\usepackage{latexsym,amsmath,amssymb,algo}
\usepackage{graphicx,epsfig}
\usepackage{url}

\newcommand{\pmin}{\textrm{pmin}}
\newcommand{\pmax}{\textrm{pmax}}
\newtheorem{observation}{Observation}{\bfseries}{\itshape}
\newcommand{\rank}{{\textit{rank}}}
\newcommand{\select}{{\textit{select}}}
\newcommand{\ff}{\textsc{firstfit}}

\begin{document}

\markboth{P.\ Burcsi, F.\ Cicalese, G.\ Fici, Zs.\ Lipt\'ak}
{Algorithms for Jumbled Pattern Matching in Strings}

\title{Algorithms for Jumbled Pattern Matching in Strings}

\author{P{\'e}ter Burcsi\inst{1} \and Ferdinando Cicalese\inst{2} 
\and Gabriele Fici\inst{3} \and Zsuzsanna Lipt{\'a}k\inst{4}}

\institute{Department of Computer Algebra,
E\"otv\"os Lor\'and University,
Hungary
\email{bupe@compalg.inf.elte.hu}
\and Dipartimento di Informatica ed Applicazioni, University of Salerno, Italy 
\email{cicalese@dia.unisa.it} 
\and I3S, UMR6070, CNRS et Universit\'{e} de Nice-Sophia Antipolis, France
\email{fici@i3s.unice.fr}
\and AG Genominformatik, Technische Fakult\"at, Bielefeld University, Germany 
\email{zsuzsa@cebitec.uni-bielefeld.de}
}

\maketitle

\begin{center}
\fbox{
\begin{minipage}{12cm}
{\em Electronic version of an article accepted for publication in the International Journal of Foundations of Computer Science (IJFCS),  to appear in  2011
\copyright World Scientific Publishing Company,} \url{www.wordscinet.com/ijfcs/}
\end{minipage}
}
\end{center}

\begin{abstract}
The Parikh vector $p(s)$ of a string $s$ over a finite ordered alphabet $\Sigma = \{a_{1},\ldots,a_{\sigma}\}$ is defined
as the  vector of multiplicities of the characters, $p(s) = (p_{1},\ldots,p_{\sigma})$, where $p_{i}=|\{j \mid
s_{j} = a_{i}\}|.$  Parikh vector $q$ occurs in $s$ if $s$ has a substring $t$ with $p(t)=q$. The problem of searching
for a query $q$ in a text $s$ of length $n$ can be solved simply and worst-case optimally with a sliding window approach
in $O(n)$ time. We present two novel algorithms for the case where the text is fixed and many queries arrive over time.

The first algorithm only {\em decides} whether a given Parikh vector appears in a binary text. It uses a linear size
data structure and decides each query in $O(1)$ time.  The preprocessing can be done trivially in $\Theta(n^2)$
time.

The second algorithm finds all  occurrences of a given Parikh vector in a text over an arbitrary alphabet of size $\sigma
\geq 2$ and has sub-linear expected time complexity. More precisely, we present two variants of the algorithm, both
using an $O(n)$ size data structure, each of which can be constructed in $O(n)$ time. The first solution is very simple
and easy to implement and leads to an expected query time of $O(n(\frac{\sigma}{\log\sigma})^{1/2}\frac{\log
m}{\sqrt{m}})$, where $m=\sum_i q_i$ is the length of a string with Parikh vector $q$. The second uses wavelet trees and
improves the expected runtime to 
$O(n(\frac{\sigma}{\log\sigma})^{1/2}\frac{1}{\sqrt{m}})$, i.e., by a factor of $\log m$. 
\end{abstract}

\begin{center}{\bf Keywords:} Parikh vectors, permuted strings, pattern matching, string algorithms, average case analysis, text indexing, non-standard string matching
\end{center}

\newpage

\section{Introduction}

Parikh vectors of strings count the multiplicity of the characters. They have been reintroduced many times by
different names (compomer~\cite{Boecker07}, composition~\cite{Benson03}, Parikh vector~\cite{Salomaa03}, 
permuted string~\cite{ButmanEL04}, permuted pattern~\cite{EresLP04}, and others). They are natural objects to
study, due to their numerous applications; for instance, in
computational biology, they have been applied in alignment algorithms~\cite{Benson03}, SNP discovery~\cite{Boecker07},
repeated pattern discovery~\cite{EresLP04}, and, most naturally, in interpretation of mass spectrometry
data~\cite{boecker04sequencing}. Parikh vectors can be seen as a generalization of strings, where we view two strings
as equivalent if one can be turned into the other by permuting its characters; in other words, if the two strings have
the same Parikh vector.

The problem we are interested in here is answering the question whether a query Parikh vector $q$ appears in a given
text $s$ (decision version), or where it occurs (occurrence version). An occurrence of $q$ is defined as an
occurrence of a substring $t$ of $s$ with Parikh vector $q$. The problem can be viewed as an approximate pattern
matching problem: We are looking for an occurrence of a jumbled version of a query string $t$, i.e.\ for the occurrence
of a substring $t'$ which has the same Parikh vector. In the following, let $n$ be the length of the text $s$, $m$
the length of the query $q$ (defined as the length of a string $t$ with Parikh vector $q$), and
$\sigma$ the size of the alphabet.

The above problem (both decision and occurrence versions) can be solved  with a simple sliding window based
algorithm, in $O(n)$ time and $O(\sigma)$ additional storage space. This is worst case optimal with respect to the case 
of one query.  However, when we expect to search for many queries in the same string, the above approach 
 leads to $O(Kn)$ runtime for $K$ queries.  To the best of our knowledge, no faster approach is known. This is
in stark contrast to the classical exact pattern matching problem, where all {\em exact} occurrences of a query
pattern of length $m$ are sought in a text of length $n$. 
In that case, for one query, any naive approach leads to $O(nm)$ runtime, while quite involved
ideas for preprocessing and searching are necessary to achieve an improved runtime of $O(n+m)$, as do the 
Knuth-Morris-Pratt \cite{KnuthMP77}, Boyer-Moore \cite{BoyerM77} and Boyer-Moore-type algorithms (see, e.g., 
\cite{ApostolicoG86,Horspool80}). However, when many queries
are expected, the text can be preprocessed to produce a data structure of size linear in $n$, such as a suffix tree,
suffix array, or suffix automaton, which then allows
to answer individual queries in time linear in the length of the pattern (see any textbook on string algorithms, 
e.g.~\cite{Smyth,Lothaire3}).

\subsection{Related work}

Jumbled pattern matching is a special case of approximate pattern matching. It has been used as a filtering step
in approximate pattern matching algorithms~\cite{JTU96}, but rarely considered in its own right.

The authors of~\cite{ButmanEL04} present an algorithm for finding all occurrences of a Parikh vector
in a runlength encoded text. The algorithm's time complexity is $O(n' + \sigma)$, where $n'$ is the length of the
runlength encoding of $s$. Obviously, if the string is not runlength encoded, a preprocessing phase of time $O(n)$ 
has to be added. However, this may still be feasible if many queries are expected. To the best of our knowledge, 
this is the only algorithm that has been presented for the problem we treat here.

An efficient algorithm for computing all Parikh fingerprints of substrings of a given string was developed
in~\cite{AmirALS03}. Parikh fingerprints are Boolean vectors where the $k$'th entry is $1$ if and only if $a_k$ appears
in the string. The algorithm involves storing a data point for each Parikh fingerprint, of which there are at most
$O(n\sigma)$ many. This approach was adapted in~\cite{EresLP04} for Parikh vectors and applied to identifying all
repeated Parikh vectors within a given length range; using it to search for queries of arbitrary length would imply
using $\Omega(P(s))$ space, where $P(s)$ denotes the number of different Parikh vectors of substrings of $s$. This is
not desirable, since, for arbitrary alphabets, there are non-trivial strings of any length with quadratic $P(s)$~\cite{CELSW04}.

\subsection{Results}

In this paper, we present two novel algorithms which perform significantly better than the simple window algorithm, in
the case where many queries arrive. 

For the binary case, we present an algorithm which answers {\em decision queries} in $O(1)$
time, using a data structure of size $O(n)$ (Interval Algorithm, Sect.~\ref{sec:binary}). 
The data structure  is  constructed  in $\Theta(n^2)$ time.

For general alphabets, we present an algorithm with expected sublinear runtime 
which uses $O(n)$ space to answer {\em occurrence queries} (Jumping Algorithm, Sect.~\ref{sec:jumping}). 
We present two different variants of the algorithm. The first one uses a very simple
data structure (an inverted table) and answers queries in time $O(\sigma J \log (\frac{n}{J} + m)),$ where $J$ denotes
the number of iterations of the main loop of  the algorithm. We then show that the expected value of $J$ for the case of
random strings and patterns is $O(\frac{n}{\sqrt{m}\sqrt{\sigma\log\sigma}})$, 
yielding an expected runtime of $O(n( \frac{\sigma}{\log\sigma})^{1/2}\frac{\log m}{\sqrt{m}})$, per query

The second variant of the algorithm uses wavelet trees~\cite{GrossiGV03} and has query time $O(\sigma J)$,
yielding an overall expected runtime of $O(n( \frac{\sigma}{\log\sigma})^{1/2}\frac{1}{\sqrt{m}})$, per query.
(Here and in the following,
$\log$ stands for logarithm base $2$.)
 
Our simulations on random strings and real biological strings confirm the 
sublinear behavior of the algorithms  in practice. This is a significant improvement over the simple window algorithm
w.r.t.\ expected runtime, both for a single query and repeated queries over one string.


The Jumping Algorithm is reminiscent of the Boyer-Moore-like approaches  to the classical exact string matching 
problem~\cite{BoyerM77,ApostolicoG86,Horspool80}. This  analogy is used both in its presentation and 
in the analysis of the number of iterations performed by the algorithm.

\section{Definitions and problem statement}

Given a finite ordered alphabet $\Sigma = \{a_1,\ldots,a_{\sigma}\}, a_1\leq \ldots\leq a_{\sigma}$. For a string $s\in
\Sigma^*$, $s=s_1\ldots s_n$, the {\em Parikh vector} $p(s) = (p_1,\ldots,p_{\sigma})$ of $s$ defines the multiplicities of the characters in
$s$, i.e.\ $p_i = |\{ j \mid s_j = a_i\}|$, for $i=1,\ldots,
\sigma$. For a Parikh vector $p$, the {\em length} $|p|$ denotes the length of a string with Parikh vector $p$, i.e.\ $|p| = \sum_i p_i$. 
An {\em occurrence} of a Parikh vector $p$ in $s$ is an occurrence of a substring $t$ with $p(t)=p$. (An occurrence of
$t$ is a pair of positions $0\leq i\leq j\leq n$, such that $s_i\ldots s_j=t$.) 
A Parikh vector that occurs in $s$ is called a sub-Parikh vector of $s$. The prefix of length
$i$ is denoted $pr(i) = pr(i,s) = s_1\ldots s_i$, and the Parikh vector of $pr(i)$ as $prv(i) = prv(i,s) = p(pr(i))$. 

For two Parikh vectors $p,q\in {\mathbb N}^{\sigma}$, we
define $p \leq q$ and $p+q$ component-wise: $p\leq q$ if and only if $p_i \leq q_i$ for all $i=1,\ldots, \sigma$, and 
$p+q = u$ where $u_i = p_i + q_i$ for $i=1,\ldots,\sigma$. Similarly, for $p\leq q$, we set $q-p = v$ where 
$v_i = q_i - p_i$ for $i=1,\ldots,\sigma$.

\begin{quote}{\bf Jumbled Pattern Matching (JPM).} Let $s\in \Sigma^*$ be given, $|s|=n$. For a Parikh vector 
$q\in {\mathbb N}^{\sigma}$ (the query), $|q| = m$, find all occurrences of $q$ in $s$. The {\em decision version} 
of the problem is where we only want to know whether $q$ occurs in $s$.
\end{quote}

We assume that $K$ many queries arrive over time, so some preprocessing may be worthwhile. 

\medskip

Note that for $K=1$, both the decision version and the occurrence version can be solved worst-case optimally with a
simple window
algorithm, which moves a fixed size window of size $m$ along string $s$. Maintain the Parikh vector $c$ of the 
current window and a counter $r$ which counts indices $i$ such that $c_i \neq q_i$. Each sliding step
costs either 0 or 2 update operations of $c$, and possibly one increment or decrement of $r$.
This algorithm solves both the decision and occurrence problems and has running time $\Theta(n)$, using additional 
storage space $\Theta(\sigma)$.

Precomputing, sorting, and storing all sub-Parikh vectors of $s$ would lead to $\Theta(n^2)$ storage space, since
there are non-trivial strings with a quadratic number of Parikh vectors over arbitrary alphabets~\cite{CELSW04}. Such
space usage is inacceptable in many applications.

For small queries, the problem can be solved exhaustively with a linear size indexing structure such as a
suffix tree, which can be searched down to length $m=|q|$ (of the substrings), yielding a solution to the decision problem in
time $O(\sigma^m)$. For finding occurrences, report all leaves in the subtrees below each match; this costs $O(M)$ time,
where $M$ is the number of occurrences of $q$ in $s$. Constructing the suffix tree
takes $O(n)$ time, so for $m = o(\log n)$, we get a total runtime of $O(n)$, since $M \leq n$ for any query $q$.


\section{Decision problem in the binary case}\label{sec:binary}

In this section, we present an algorithm for strings over a binary alphabet which, once a data structure of size
$O(n)$ has been constructed, answers decision queries in constant time. It makes use of the following nice property of
binary strings. 

\begin{lemma}\label{lemma:continuous}
Let $s\in \{a,b\}^*$ with $|s|=n$. Fix $1\leq m\leq n$. 
If the Parikh vectors $(x_1,m-x_1)$ and $(x_2,m-x_2)$ both occur in $s$, then so does $(y,m-y)$ for any $x_1\leq y \leq
x_2$. 
\end{lemma}

\begin{proof}
Consider a sliding window of fixed size $m$ moving along the string and let $(p_1,p_2)$ be the Parikh vector of the
current substring. When the window is shifted by one, the Parikh vector either remains unchanged (if the character
falling out is the same as the character coming in), or it becomes $(p_1+1,p_2 -1)$ resp.\ $(p_1-1,p_2 +1)$ (if they are
different). Thus the Parikh vectors of substrings of $s$ of length $m$ build a set of the form $\{(x,m-x) \mid
x=\pmin(m),\pmin(m)+1,\ldots, \pmax(m)\}$ for appropriate $\pmin(m)$ and $\pmax(m)$. \hfill \qed
\end{proof}

Assume that the algorithm has access to the values $\pmin(m)$ and $\pmax(m)$ for $m=1,\ldots,n$; then, when a query
$q = (x,y)$ arrives, it answers {\sc yes} if and only if $x \in [\pmin(x+y), \pmax(x+y)]$. The query time is $O(1)$. 

The table of the values $\pmin(m)$ and $\pmax(m)$ can be easily computed in a preprocessing step in time $\Theta(n^2)$
by scanning the string with a window of size $m$, for each $m$. Alternatively, lazy computation of the table is
feasible, since for any query $q$, only the entry $m=|q|$ is necessary. Therefore, it can be computed on the fly as
queries arrive. Then, any query will take time $O(1)$ (if the appropriate entry has already been computed), or $O(n)$
(if it has not). After $n$ queries of the latter kind, the table is completed, and all subsequent queries can be
answered in $O(1)$ time. If we assume that the query lengths are uniformly distributed, then this can be viewed as a
coupon collector problem where the coupon collector has to collect one copy of each length $m$. Then the expected
number of queries needed before having seen all $n$ coupons is $nH_n \approx n\ln n$ (see e.g.~\cite{Feller}). The
algorithm will have taken $O(n^2)$ time to answer these $n\ln n$ queries.

The assumption of the uniform length distribution may not be very realistic; however, even if it does not hold, we never
take more time than $O(n^2 + K)$ for $K$ many queries. Since any one query may take at most $O(n)$ time, our algorithm
never performs worse than the simple window algorithm. Moreover, for those queries where the table entries have to be computed, we can even run the
simple window
algorithm itself and report all occurrences, as well. For all others, we only give decision answers, but in constant
time. 

The size of the data structure is $2n = O(n)$. The overall running time for either variant is $\Theta(K+n^2)$. As soon as
the number of queries is $K=\omega(n)$, 
both variants outperform the simple window algorithm, whose running time is $\Theta(Kn)$.

\medskip

\begin{example} Let  $s = ababbaabaabbbaaabbab$. In Table~\ref{tab:val}, we give the table of $\pmin$ and
$\pmax$ for $s$. This example shows that the locality of \pmin\ and \pmax\ is preserved only in adjacent levels. As an
example, the value $\pmax(3)=3$ corresponds to the substring $aaa$ appearing only at position $14$, while $\pmax(5)=4$
corresponds to the substring $aabaa$ appearing only at position $6$.

\begin{table}[h]
\centering  \caption{An example of the linear data structure for answering queries in constant time. \label{tab:val}}
\begin{small}
\begin{raggedright}
\vspace{4mm}

\begin{tabular}{c *{30}{@{\hspace{2.4mm}}r}}
 $m$  & 1\hspace{1ex} & 2\hspace{1ex} & 3\hspace{1ex} & 4\hspace{1ex} & 5\hspace{1ex} & 6\hspace{1ex} & 7\hspace{1ex} &
8\hspace{1ex} & 9\hspace{1ex} & 10 & 11 & 12 & 13 & 14 & 15 & 16 & 17 & 18 & 19 & 20 \\
\hline \rule[-6pt]{0pt}{22pt}
$\pmin$   & 0& 0& 0& 1& 2& 2& 3& 3& 4& 4& 5& 5& 6& 7& 7& 8& 8& 9& 9& 10&\\
\hline \rule[-6pt]{0pt}{22pt}
$\pmax$   & 1& 2& 3& 3& 4& 4& 4& 5& 5& 6& 7& 7& 7& 8& 8& 9& 9& 9& 10& 10&  \\
\hline
\end{tabular}
\label{table:example-intervals}
\end{raggedright}
\end{small}
\end{table}

\end{example}


\section{The Jumping Algorithm}\label{sec:jumping}

In this section, we introduce our algorithm for general alphabets. We first give the main algorithm and then present
two different implementations of it. The first one, an inverted prefix table, is very easy to
understand and to implement, takes $O(n)$ space and $O(n)$ time to construct (both with constant $1$), and can replace
the string. Then we show how to use a wavelet tree of $s$ to implement our algorithm, which has the same space
requirements as the inverted table, can be constructed in $O(n)$ time, and improves
the query time by a factor of $\log m$.

\subsection{Main algorithm}

Let $s = s_1\ldots s_n \in \Sigma^*$ be given. Recall that $prv(i)$ denotes the Parikh vector of the prefix of $s$
of length $i$, for $i=0,\ldots, n$, where $prv(0) = p(\epsilon) = (0,\ldots,0)$. Consider Parikh vector $p\in {\mathbb N}^{\sigma}$, $p \neq
(0,\ldots,0)$. We make the following simple observations:

\begin{observation}\label{obs:pr}

\begin{enumerate}
\item For any $0 \leq i\leq j\leq n$, $p = prv(j) - prv(i)$ 
if and only if $p$ occurs in $s$ at position $(i+1,j)$.
\item If an occurrence of $p$ ends in position $j$, then $prv(j) \geq p$.
\end{enumerate}
\end{observation}

The algorithm moves two pointers $L$ and $R$ along the
text, pointing at these potential positions $i$ and $j$. Instead of moving linearly, however, the pointers
are updated in jumps, alternating between updates of $R$ and $L$, in such a manner that many positions are skipped.
Moreover, because of the way we update
the pointers, after any update it suffices to check whether $R-L = |q|$ to confirm that an occurrence has been found
(cf.\ Lemma~\ref{lemma:invariants} below). 

We first need to define a function $\textsc{firstfit}$, which returns the smallest potential position where an occurence
of a Parikh vector can end. Let $p\in {\mathbb N}^{\sigma}$, then

\begin{equation*}
\textsc{firstfit}(p) := \min\{ j \mid prv(j) \geq p\},
\end{equation*}

\noindent
and set $\textsc{firstfit}(p) = \infty$ if no such $j$ exists. We use the following rules for updating the two
pointers, illustrated in Fig.~\ref{fig:jumping}.

\begin{figure}
\begin{center}
\includegraphics[scale=0.5]{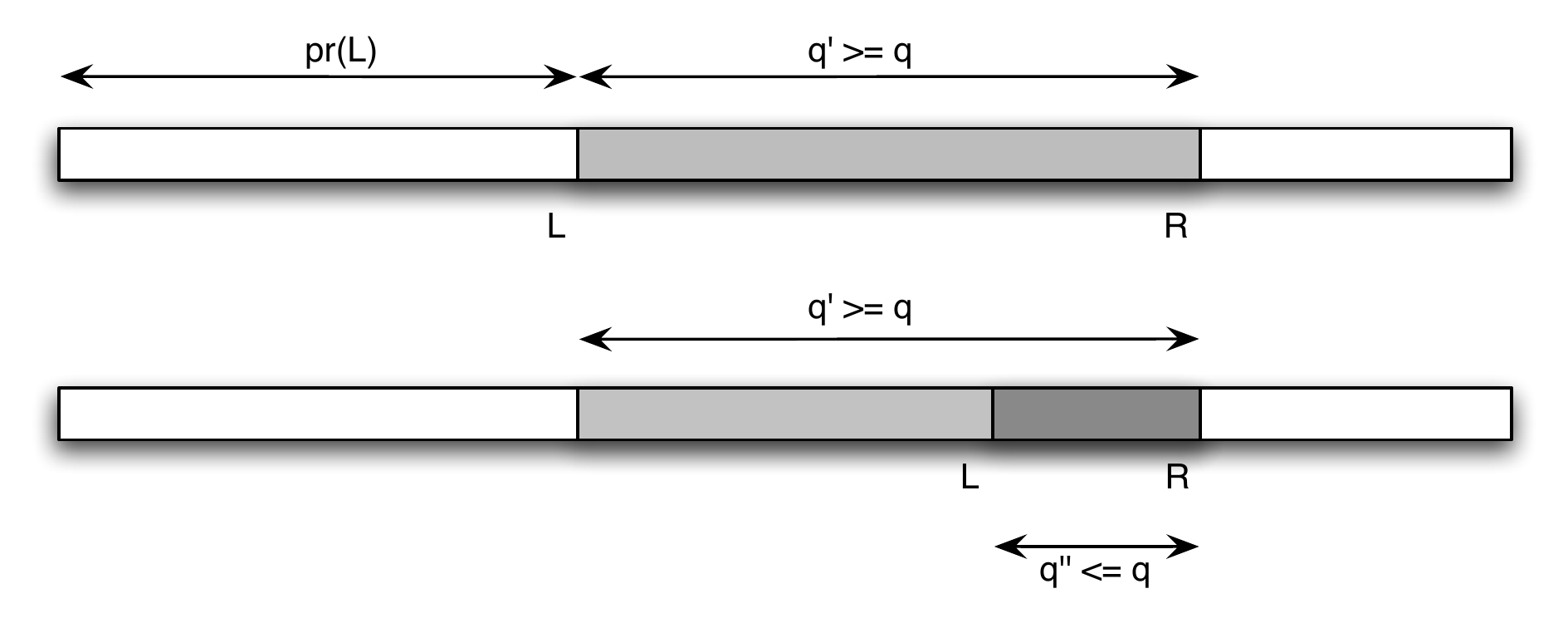}
\caption{The situation after the update of $R$ (above) and after the update of $L$ (below). $R$ is placed at the first
fit of $prv(L)+q$, thus $q'$ is a super-Parikh vector of $q$. Then $L$ is placed at the beginning of the longest
good suffix ending in $R$, so $q''$ is a sub-Parikh vector of $q$.\label{fig:jumping}}
\end{center}
\end{figure}

{\em Updating $R$:} Assume that the left pointer is pointing at position $L$, i.e.\ no unreported occurrence starts
before $L+1$.  Notice that, if there is an occurrence of $q$ ending at any position $j>L$, it must hold
that $prv(L) + q \leq prv(j)$. In other words, we must fit both $prv(L)$ and $q$ at position $j$, so we update $R$ to

\begin{equation*} R \gets \textsc{firstfit}(prv(L) + q).
\end{equation*}

{\em Updating $L$:} 
Assume that $R$ has just been updated. Thus, $prv(R) - prv(L) \geq q$ by definition of $\ff$. If equality holds, then we
have found an occurrence of $q$ in position $(L+1,R)$, and $L$ can be incremented by $1$. Otherwise $prv(R) - prv(L) > q$, 
which implies that,
interspersed between the characters that belong to $q$, there are some ``superfluous" characters. Now the first
position where an occurrence of $q$ can start is at the beginning of a {\em contiguous} sequence of characters
ending in $R$ which all belong to $q$. In other words, we need the beginning of the longest suffix of $s[L+1,R]$
with Parikh vector $\leq q$, i.e.\ the smallest position $i$ such that $prv(R) - prv(i) \leq q$, or, equivalently, 
$prv(i) \geq prv(R) - q$. Thus we update $L$ to 
 
\begin{equation*}
  L \gets \textsc{firstfit}(prv(R) - q).
\end{equation*}

Finally, in order to check whether we have found an occurrence of query $q$, after each update of $R$ or $L$, we check
whether $R-L = |q|$. In Figure~\ref{fig:jumping_pseudocode}, we give the pseudocode of the algorithm.

\begin{figure}
\begin{algorithm}{Jumping Algorithm}{
\label{algo:jumping}
\qinput{query Parikh vector $q$}
\qoutput{A set $Occ$ containing all beginning positions of occurrences of $q$ in $s$} 
}
set $m \qlet |q|; Occ \qlet \emptyset$; $L \qlet 0$;\\
\qwhile  $L < n-m$ \\
\qdo $R \qlet \textsc{firstfit}(prv(L) + q)$; \\
\qif $R-L = m$  \\ 
\qthen  add $L+1$ to $Occ$;\\
$L \qlet L+1$;  \\
\qelse $L \qlet \textsc{firstfit}(prv(R) - q)$; \\ 
\qif $R-L = m$  \\ 
\qthen 
add $L+1$ to $Occ$;\\
$L\qlet L+1$; 
\qfi 
\qfi
\qend \\ 
\qreturn $Occ$;
\end{algorithm}
\caption{Pseudocode of Jumping Algorithm\label{fig:jumping_pseudocode}}
\end{figure}

\medskip

It remains to see how to compute the $\textsc{firstfit}$ and $prv$ functions. We first prove that the algorithm is
correct. For this, we will need the following lemma.

\begin{lemma}\label{lemma:invariants}
The following algorithm invariants hold:

\begin{enumerate}
\item After each update of $R$, we have $prv(R) - prv(L) \geq q$.
\item After each update of $L$, we have $prv(R) - prv(L) \leq q$.
\item $L \leq R$.
\end{enumerate}
\end{lemma}

\begin{proof}
{\em 1.} follows directly from the definition of $\textsc{firstfit}$ and the update rule for $R$. For {\em 2.}, if an occurrence was
found at $(i,j)$, then before the update we have $L=i-1$ and $R=j$. Now $L$ is incremented by $1$, so $L=i$ and $prv(R) -
prv(L)
= q - e_{s_i} < q$, where $e_k$ is the $k$'th unity vector. Otherwise, $L \gets \textsc{firstfit}(prv(R) - q)$, and again the claim
follows directly from the definition of $\textsc{firstfit}$. For {\em 3.}, if an occurrence was found, then $L$ is incremented by $1$,
and $R-L = m-1\geq 0$. Otherwise, $L = \textsc{firstfit}(prv(R) - q) = \min \{ \ell \mid prv(\ell) \geq prv(R) - q\} \leq R$. \hfill
\qed
\end{proof}

\begin{theorem}\label{thm:jumping_correctness}
Algorithm Jumping Algorithm is correct.
\end{theorem}

\begin{proof} We have to show that (1) if the algorithm reports an occurrence, then it is correct, and (2) if there is an
occurrence, then the algorithm will find it. 

{\em (1)} If the algorithm reports an index $i$, then $(i, i+m-1)$ is an occurrence of $q$: An index $i$ is
added to $Occ$ whenever $R-L = m$. If the last update was that of $R$, then we have $prv(R) - prv(L) \geq q$ by Lemma
\ref{lemma:invariants}, and
together with $R-L = m = |q|$, this implies $prv(R) - prv(L) = q$, thus $(L+1,R) = (i,i+m-1)$ is an occurrence of $q$.
If the last update was $L$, then $prv(R) - prv(L) \leq q$, and it follows analogously that $prv(R)-prv(L) = q$.

\medskip

{\em (2)} All occurrences of $q$ are reported: Let's assume otherwise. Then
there is a minimal $i$ and $j=i+m-1$ such that $p(s[i,j])=q$ but $i$ is not reported by the algorithm. 
By Observation \ref{obs:pr}, we have $prv(j) - prv(i-1) = q$.

Let's refer to the values of $L$ and $R$ as two sequences $(L_{k})_{k=1,2,\ldots}$ and $(R_{k})_{k=1,2,\ldots}$.
So we have $L_1 = 0$, and for all $k\geq 1,$ $R_{k} = \textsc{firstfit}(prv(L_{k}) + q)$, and $L_{k +1} = L_{k}+1$ if
$R_{k} - L_{k} = m$ and $L_{k +1} = \textsc{firstfit}(prv(R_{k}) - q)$ otherwise. In particular, $L_{k+1} > L_k$ for all $k$.

First observe that if for some $k$, $L_{k} = i-1$, then $R$ will be updated to $j$ in the next step, and we are done.
This is because $R_{k} = \textsc{firstfit}(prv(L_{k}) + q) = \textsc{firstfit}(prv(i-1) + q) = \textsc{firstfit}(prv(j)) = j$. Similarly, if for some $k$,
$R_{k} = j$, then we have $L_{k +1} = i-1$.

So there must be a $k$ such that $L_{k} < i-1 < L_{k +1}$. Now look at $R_k$. Since there is an occurrence of $q$
after $L_{k}$ ending in $j$, this implies that $R_{k} = \textsc{firstfit}(prv(L_{k}) + q) \leq j$. However, we cannot have
$R_{k} =j$, so it follows that $R_{k} < j$. On the other hand,  $i-1 < L_{k +1} \leq R_{k}$ by our
assumption and by Lemma \ref{lemma:invariants}. So $R_k$ is pointing to a position somewhere between $i-1$ and $j$,
i.e.\ to a position within our occurrence of $q$. Denote the remaining part of $q$ to the right of $R_k$ by $q'$: $q'
= prv(j) - prv(R_k)$. Since $R_k = \textsc{firstfit}(prv(L_k) + q)$, all characters of $q$ must fit between $L_k$ and $R_k$, so
the Parikh vector $p = prv(i) - prv(L_k)$ is a super-Parikh vector of $q'$. If $p = q'$, then there is an occurrence of
$q$ at $(L_k + 1, R_k)$, and by minimality of $(i,j)$, this occurrence was correctly identified by the algorithm.
Thus, $L_{k+1} = L_k +1 \leq i-1$, contradicting our choice of $k$. It follows that $p > q'$ and we have to find the 
longest good
suffix of the substring ending in $R_k$ for the next update $L_{k+1}$ of $L$. But $s[i,R_k]$ is a good suffix because
its Parikh vector is a sub-Parikh vector of $q$, so $L_{k+1} = \textsc{firstfit}(prv(R_k) - q) \leq i-1$, again in contradiction to
$L_{k+1} > i-1$.\hfill \qed
\end{proof}

We illustrate the proof in Fig.~\ref{fig:proof_of_correctness}.

\begin{figure}
\begin{center}
\includegraphics[scale=0.5]{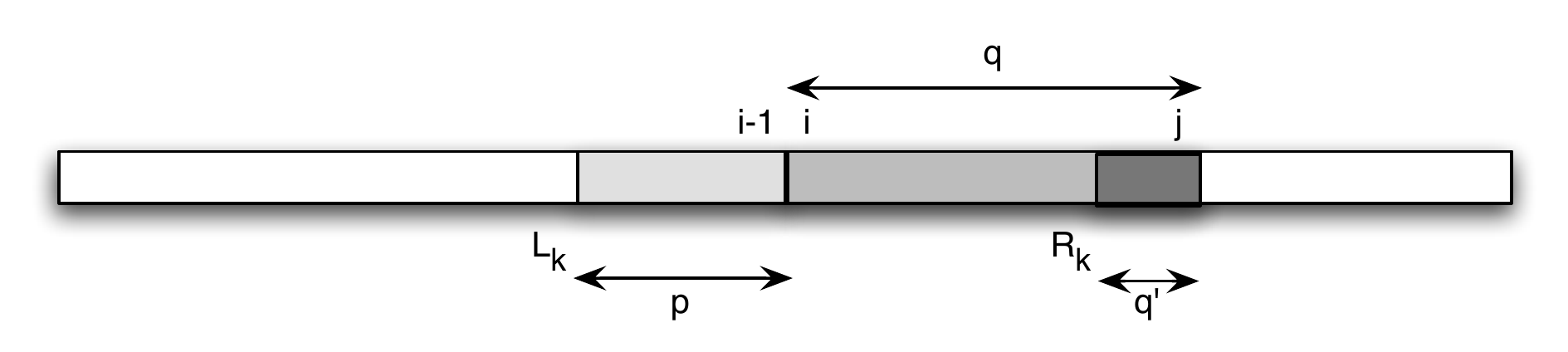}
\caption{\label{fig:proof_of_correctness} Illustration for proof of correctness.}
\end{center}
\end{figure}

\subsection{Variant using an inverted table}\label{sec:firstfit}

Storing all prefix vectors of $s$ would require $O(\sigma n)$ storage space, which may be too much. Instead, we
construct an ``inverted prefix vector table" $I$ containing the increment positions of the prefix vectors:
for each character $a_k\in \Sigma$, and each value $j$ up to $p(s)_k$, the position in $s$ of the $j$'th occurrence
of character $a_k$. Formally, $I[k][j] = \min \{ i \mid prv(i)_k \geq j \}$ for $j\geq 1$, and $I[k][0]=0$. 
Then we have 

\begin{equation*}
\textsc{firstfit}(p) = \max _{k=1,\ldots,\sigma} \{ I[k][p_k]\}.
\end{equation*}

We can also compute the prefix vectors $prv(i)$ from table $I$: For $k=1,\ldots,\sigma$,

\begin{equation*}
prv(j)_k = \max \{ i \mid I[k][i] \leq j \}
\end{equation*}

The obvious way to find these values is to do binary search for $j$ in each row of $I$. However, this would take time 
$\Theta(\sigma \log n)$; a better way is to use information already
acquired during the run of the algorithm. By Lemma \ref{lemma:invariants}, it always holds that
$L\leq R$. Thus, for computing $prv(R)_k$, it suffices to search for $R$ between $prv(L)_k$ and $prv(L)_k + (R-L)$. 
This search takes time proportional to $\log(R-L)$. Moreover, after each update of $L$, we have $L\geq R - m$, so 
when computing $prv(L)_k$, we can restrict the search for $L$ to between $prv(R)_k-m$ and $prv(R)_k$, in time $O(\log
m)$. For more details, see Section~\ref{sec:jumping_analysis}.

\medskip

Table $I$ can be computed in one pass over $s$ (where we take the liberty of identifying character
$a_k\in\Sigma$ with its index $k$). The variables $c_{k}$ count the number of occurrences of character $a_k$ seen so
far, and are initialized to $0$.

\begin{algorithm}{Construct $I$}{
\label{algo:preproc}}
\qfor $i=1$ \qto $n$ \\
$c_{s_i} = c_{s_i}+1$;\\
$I[s_i][c_{s_i}] = i;$ 
\qend
\end{algorithm}

Table $I$ requires $O(n)$ storage space (with constant 1). Moreover, the string $s$ can be discarded, so we
have zero additional storage. (Access to $s_i, 1\leq i \leq n,$ is still possible, at cost $O(\sigma \log n)$.)

\medskip

\begin{example} Let $\Sigma = \{a,b,c\}$ and $s=cabcccaaabccbaacca$. The prefix vectors of $s$ are given below. Note
that the algorithm does not actually compute these.

\bigskip
\noindent
$
\begin{array}{p{.9cm} | @{\hspace{.1cm}}*{19}{p{.53cm}}}
\text{pos.} && 1 & 2 & 3 & 4 & 5 & 6 & 7 & 8 & 9 & 10 & 11 & 12 & 13 & 14 & 15 & 16 & 17 & 18\\
\hline
$s$  && $c$ & $a$ & $b$ & $c$ & $c$ & $c$ & $a$ & $a$ & $a$ & $b$ & $c$ & $c$ & $b$ & $a$ & $a$ & $c$ & $c$ & $a$\\
\# $a$'s & 0 & 0 & 1 & 1 & 1 & 1 & 1 & 2 & 3 & 4 & 4 & 4 & 4 & 4 & 5 & 6 & 6 & 6 & 7\\
\# $b$'s & 0 & 0 & 0 & 1 & 1 & 1 & 1 & 1 & 1 & 1 & 2 & 2 & 2 & 3 & 3 & 3 & 3 & 3 & 3\\
\# $c$'s & 0 & 1 & 1 & 1 & 2 & 3 & 4 & 4 & 4 & 4 & 4 & 5 & 6 & 6 & 6 & 6 & 7 & 8 & 8\\
\end{array}
$

\bigskip

\noindent The inverted prefix table $I$:

\medskip

\noindent
$
\begin{array}{ p{.8cm} |@{\hspace{.2cm}} *{9}{p{0.63cm}}}
 & 0 & 1 & 2 & 3 & 4 & 5 & 6 & 7 & 8\\
 \hline
 a & 0 & 2 & 7 & 8 & 9 & 14 & 15 & 18 \\
 b & 0 & 3 & 10 & 13 \\
 c & 0 & 1 & 4 & 5 & 6 & 11 & 12 & 16 & 17\\
\end{array}
$

\bigskip

\noindent Query $q=(3,1,2)$ has 4 occurrences, beginning in positions $5,6,7,13$, since $(3,1,2) = prv(10) -
prv(4) = prv(11) - prv(5) = prv(12) - prv(6) = prv(18) - prv(12)$. The values of $L$ and $R$ are given below:

\bigskip

\noindent
\begin{tabular}{r|@{\hspace{.2cm}}*{7}{p{.5cm}}}
$k$, see proof of Thm.~\ref{thm:jumping_correctness} & 1 & 2 & 3 & 4 & 5 & 6 & 7\\
\hline
L & 0 & 4 & 5 & 6 & 7 & 10 & 12\\
R & 8 & 10 & 11 & 12 & 14 & 18 & 18\\
\hline
occurrence found? & -- & yes & yes & yes & -- & -- & yes\\
\end{tabular}

\end{example}

\subsection{Variant using a wavelet tree}

A wavelet tree on $s\in \Sigma^*$ allows {\em rank, select,} and {\em access} queries in time $O(\log
\sigma)$. For $a_k\in\Sigma$, $\rank_k(s,i) = |\{ j \mid s_j=a_k, j \leq i\}|$, the number of occurrences of
character $a_k$ up to and including position $i$, while $\select_k(s,i) = \min \{ j \mid \rank_k(s,j) \geq i\}$,
the position of the $i$'th occurrence of character $a_k$. When the string is clear, we just use $\rank_k(i)$ 
and $\select_k(i)$. Notice that 

\begin{itemize}
\item $prv(j) = (\rank_1(j),\ldots, \rank_{\sigma}(j))$, and 
\item for a Parikh vector $p = (p_1,\ldots, p_{\sigma})$, $\ff(p) = \max_{k=1,\ldots,\sigma}\{\select_k(p_k)\}$.
\end{itemize}

So we can use a wavelet tree of string $s$ to implement those two functions. We give a brief recap of wavelet
trees, and then explain how to implement the two functions above in $O(\sigma)$ time each.

\medskip

A wavelet tree is a complete binary tree with $\sigma = |\Sigma|$ many leaves. To each inner node, a bitstring is
associated which is defined recursively, starting from the root, in the following way. If $|\Sigma| = 1$,
then there is nothing to do (in this case, we have reached a leaf). Else split the alphabet into two roughly equal 
parts, $\Sigma_{\rm left}$ and $\Sigma_{\rm right}$. Now construct a bitstring of length $n$ from $s$ by replacing each
occurrence of a character $a$ by $0$ if $a\in\Sigma_{\rm left}$, and by $1$ if $a\in \Sigma_{\rm right}$. Let $s_{\rm
left}$ be the subsequence of $s$ consisting only of characters from $\Sigma_{\rm left}$, and $s_{\rm right}$ that
consisting only of characters from $\Sigma_{\rm right}$. Now recurse on the
left child with string $s_{\rm left}$ and alphabet $\Sigma_{\rm left}$, and on the right child with $s_{\rm right}$
and $\Sigma_{\rm right}$. An illustration is given in Fig.~\ref{fig:wavelet}.
At each inner node, in addition to the bitstring $B$, we have a data structure of size $o(|B|)$, which
allows to perform  $\rank$ and $\select$ queries on bit vectors in constant time 
(\cite{Munro96,Clark96,NavMaek07}). 

\medskip

Now, using the wavelet tree of $s$, any \rank\ or \select\ operation on $s$ takes time $O(\log \sigma)$, which would
yield $O(\sigma\log \sigma)$ time for both $prv(j)$ and $\ff(p)$. However, we can implement both in a way that they need
only $O(\sigma)$ time: In order to compute $\rank_k(j)$, the wavelet tree, which has $\log \sigma$ levels, has to be
descended from the root to leaf $k$. Since for $prv(j)$, we need all values $\rank_1(j),\ldots, \rank_{\sigma}(j)$
simultaneously, we traverse the complete tree in $O(\sigma)$ time. 

For computing $\ff(p)$, we need $\max_k\{\select_k(p_k)\}$, which can be computed bottom-up
in the following way. We define a value $x_u$ for each node $u$. If $u$ is a leaf, then $u$ corresponds to some character
$a_k\in \Sigma$; set $x_u=p_k$. For an inner node $u$, let $B_u$ be the bitstring at $u$. We define $x_u$ by 
$x_u = \max\{\select_0(B_u,x_{\rm left}),$ $\select_1(B_u,x_{\rm right})\}$, where $x_{\rm left}$ and $x_{\rm right}$ are the values already computed for the left resp.\ right child of $u$. The desired value is
equal to $x_{\rm root}$.

\medskip

\begin{example} Let $s=bbacaccabaddabccaaac$ (cp.\ Fig.~\ref{fig:wavelet}). We 
demonstrate the computation of $\ff(2,3,2,1)$ using the wavelet tree. We have $\ff(2,3,2,1)$ $= \max\{\select_a(s,2),
\select_b(s,3),
\select_c(s,2),
\select_d(s,1)\}$, where in slight
abuse of notation we put the character in the subscript instead of its number. Denote the bottom left bitstring as
$B_{a,b}$, the bottom right one as $B_{c,d}$, and the top bitstring as  $B_{a,b,c,d}$. Then
we get $\max\{\select_0(B_{a,b},2), \select_1(B_{a,b},3)\} = \max\{4,6\}=6$, and 
$\max\{\select_0(B_{c,d},2), \select_1(B_{c,d},1)\} = \max \{2,4\} = 4$. So at the next level, we compute 
$\max\{\select_0(B_{a,b,c,d},6),\select_1(B_{a,b,c,d},4)\} = \max\{9,11\} = 11$.
\end{example}

\begin{figure}
\begin{center}
\includegraphics[scale=0.5]{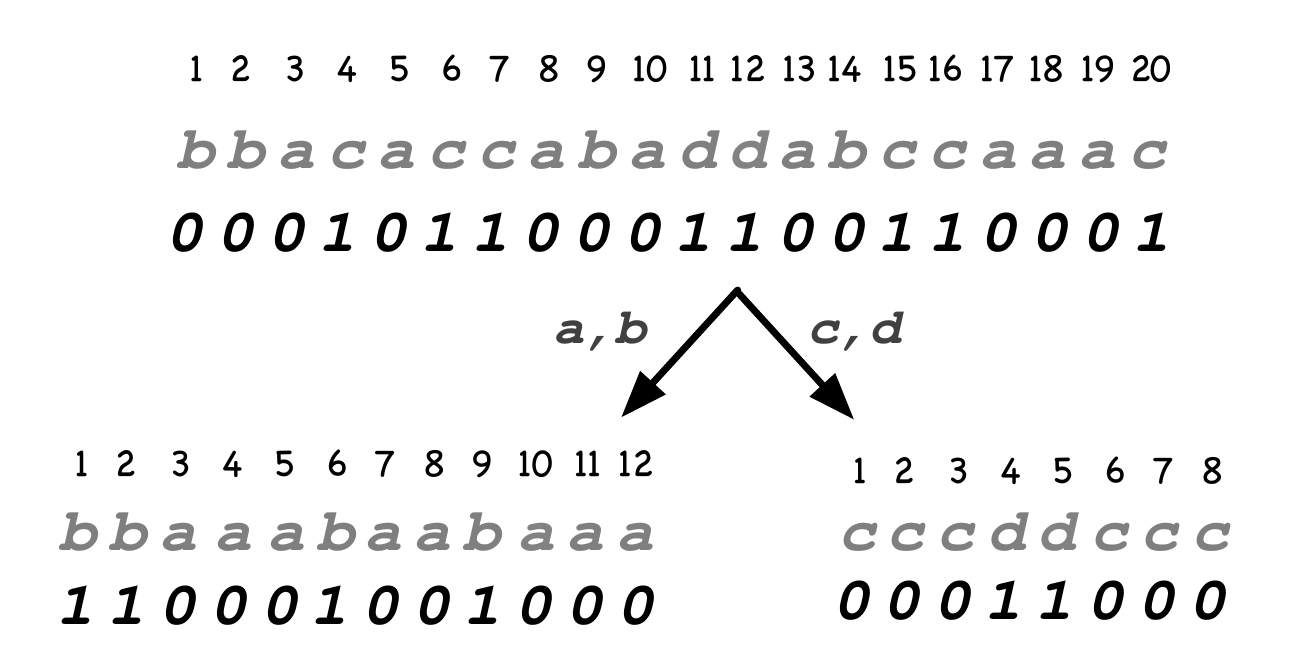}
\caption{The wavelet tree for string $bbacaccabaddabccaaac$. For clarity, the leaves have been omitted. Note
also that the third line at each inner node (the strings over the alphabet $\{a,b,c,d\}$) are only included for illustration.
\label{fig:wavelet}}
\end{center}
\end{figure}


\subsection{Algorithm Analysis}\label{sec:jumping_analysis}

Let $\mathbb{A}_1(s,q)$ denote the running time of the Jumping Algorithm using inverted tables over a text $s$ and a
Parikh vector $q$, and $\mathbb{A}_2(s,q)$ that of the Jumping Algorithm using a wavelet tree.
Further, let $J = J(s,q)$ be the number of iterations performed in the {\tt while} loop in line 2, i.e., 
the number of jumps performed by the algorithm on the input $q.$

The time spent in each iteration depends on how the functions $\ff$ and $prv$ are
implemented (lines 3 and 7). In the wavelet tree implementation, as we saw before, both take time $O(\sigma)$, so the
overall runtime of the algorithm is 

\begin{equation}\label{eq:jump_complx2}
\mathbb{A}_2(s,q) = O(\sigma J).
\end{equation}

\medskip

For the inverted table implementation, it is easy to see that computing $\textsc{firstfit}$ takes $O(\sigma)$ time.  
Now denote, for each $i=1, \dots, J,$ by $\hat{L}_i, \, \hat{R}_i$ the value
of $L$ and $R$ after the $i$'th execution of line 3 of the algorithm,
respectively.\footnote{The $\hat{L}_i$ and $\hat{R}_i$ coincide
with the $L_k$ and $R_k$ from the proof of Theorem~\ref{thm:jumping_correctness} almost but not completely: When an
occurrence
is found after the update of $L$, then the corresponding pair $L_k,R_k$ is skipped here. The reason is that now we are
only considering those updates that carry a computational cost.}
The computation of $prv(\hat{L_i})$ in line 3 takes $O(\sigma \log m)$: For each $k=1, \dots, \sigma,$ the component
$prv(\hat{L_i})_k$  can be determined by binary search over the 
list $I[k][prv(\hat{R}_{i-1})_k - m], I[k][prv(\hat{R}_{i-1})_k - m+1], \dots, I[k][prv(\hat{R}_{i-1})_k].$ By   $\hat{L_i} \geq \hat{R}_{i-1} - m,$ the claim follows.

The computation of $prv(\hat{R_i})$ in line 7 takes $O(\sigma \log (\hat{R_i} - \hat{R}_{i-1} + m)).$ 
Simply observe that in the  prefix ending at position $\hat{R_i}$ there can 
be at most $\hat{R_i} - \hat{L_i}$ more occurrences of the $k$'th character than there are in the prefix ending 
at position $\hat{L_i}.$ Therefore, as before, 
we can determine  $prv(\hat{R_i})_k$ by binary search over the list 
$I[k][prv(\hat{L}_{i})_k], I[k][prv(\hat{L}_{i})_k +1], \dots, I[k][prv(\hat{L}_{i})_k + \hat{R_i} - \hat{L_i}].$  
Using the fact that $\hat{L_i} \geq \hat{R}_{i-1} - m,$ 
the desired  bound   follows. 

The last three observations imply 
$$\mathbb{A}_1(s,q) = 
O\left(\sigma  J \log m +  \sigma \sum_{i=1}^J   \log (\hat{R_i} - \hat{R}_{i-1} + m)  \right).$$ 

Notice that this is an overestimate, since line 7 is only executed if no occurrence was found after the current update
of $R$ (line 4). Standard algebraic manipulations using Jensen's inequality (see, e.g.~\cite{Jukna98}) yield 
$\sum_{i=1}^J   \log (\hat{R_i} - \hat{R}_{i-1} + m) \leq J \log \left(\frac{n}{J} + m \right).$ 
Therefore we obtain 

\begin{equation} \label{eq:jump_complx1}
\mathbb{A}_1(s,q) = O\left(\sigma J \log\left(\frac{n}{J} + m \right) \right).
\end{equation}

\subsubsection{Average case analysis of $J$}

The worst case running time of the Jumping Algorithm, in either implementation, is superlinear, since there exist strings
$s$ of any length $n$ and Parikh
vectors $q$ such that $J = \Theta(n)$: For instance, on the string $s=ababab\ldots ab$ and $q=(2,0)$, the algorithm
will execute $n/2$ jumps. 

This sharply contrasts with the experimental evaluation we present later. The Jumping Algorithm appears to have
in practice 
a sublinear behavior. In the rest of this section we provide an average case analysis of the running time of the Jumping
Algorithm leading to the conclusion that its expected running time is sublinear.  

\medskip

We assume that the string $s$ is given as a sequence of i.i.d.\ random variables uniformly distributed over the
alphabet $\Sigma.$ According to Knuth {\em et al.\ }~\cite{KnuthMP77}
``It might be argued that the average case taken over random strings is of little interest, since a user rarely searches for a 
random string. However, this model is a reasonable approximation when we consider those pieces of text that do not contain 
the pattern [\dots]''. The  experimental results we provide will show that this is indeed the case.

\medskip

Let us concentrate on the behaviour of the algorithm when scanning a (piece of the) 
string which does not contain a match. According to the above observation we 
can reasonably take this as a measure of the performance of the algorithm, considering that 
for any match found there is an additional step of size 1, which we can charge as the cost of the 
output.

Let $E_{m, \sigma}$ denote the expected value of the distance between $R$ and $L$, following an update of $R,$ i.e. if
$L$ is in position $i,$ then we are interested in the value 
$\ell$ such that $\textsc{firstfit}(prv(i) + q) = i+\ell.$ Notice that the probabilistic assumptions made  on the string,
together with the assumption of absence of matches, allows us to treat this value as independent of the position $i.$ We
will show the following result about $E_{m,\sigma}.$ For the sake of the clarity, we defer the proof of this technical
fact to the next section.
 
\begin{lemma} \label{lemma:avg_jump}
$E_{m, \sigma}  = \Omega\left(m + \sqrt{m \sigma \ln \sigma}\right).$
 \end{lemma}
 
At each iteration (when there is no match) the $L$ pointer is moved forward
to the farthest position from $R$  such that the Parikh vector of the substring between $L$ and $R$ is 
a sub-Parikh vector of $q.$ In particular, we can upper bound the distance between the 
new positions of $L$ and $R$ with $m.$  Thus  for the expected number of jumps performed by the 
algorithm, measured as the average number of times we move $L$, we have 
\begin{equation} \label{eq:avg_J}
\mathbb{E}[J] = \frac{n}{E_{m,\sigma} - m} = O\left(\frac{n}{\sqrt{m \sigma \ln \sigma}} \right).
\end{equation}
 
Recalling~\eqref{eq:jump_complx2} and ~\eqref{eq:jump_complx1}, and using~\eqref{eq:avg_J} 
for  a random instance we have the following result concerning the average case complexity of the Jumping Algorithm.

\begin{theorem}\label{thm:time_jumping}
Let $s\in \Sigma^*$ be fixed. Algorithm Jumping Algorithm finds all occurrences of a query $q$ 
\begin{enumerate}
\item in expected time $O(n(\frac{\sigma}{\log \sigma})^{1/2}\frac{\log m}{\sqrt m})$ using an
 inverted prefix table of size $O(n)$, which can be constructed in a preprocessing step in time $O(n)$;
\item in expected time $O(n(\frac{\sigma}{\log \sigma})^{1/2}\frac{1}{\sqrt{m}})$ using a wavelet 
tree of $s$ of size $O(n)$, which can be computed in a preprocessing step in time $O(n)$.
\end{enumerate}
\end{theorem}

We conclude this section by remarking once more that the above estimate obtained by the approximating probabilistic
automaton appears to be confirmed by the experiments. 

\subsubsection{The proof of Lemma \ref{lemma:avg_jump}}

We shall argue asymptotically with $m$ and according to whether or not the Parikh vector $q$ is balanced, and in the
latter case according to its degree of {\em unbalancedness}, measured as the magnitude of its largest and smallest components. 
\medskip

\noindent
{\em Case 1.}  $q$ is  balanced, i.e., 
$q = (\frac{m}{\sigma}, \dots, \frac{m}{\sigma}).$
Then, from equations (7) and (12) of \cite{May08}, it follows that  
\begin{equation} \label{eq:coupcoll}
E_{m, \sigma} \approx  m + \begin{cases} m2^{-m}{m \choose {m/2}} & \mbox{if } \sigma = 2, \cr
\sqrt{2m \sigma \ln \frac{\sigma}{\sqrt{2\pi}}} & \mbox{otherwise.}
\end{cases}
\end{equation}

The author of \cite{May08}  studied a variant of the well known coupon collector
problem in which the collector has to accumulate a certain number of copies of each coupon. It should not be hard 
to see that by  identifying the characters with the coupon types, the random string with the sequence of coupons obtained,
and the query Parikh vector with the number of copies we require for each coupon type, 
the expected time when the collection is finished is the same as our  $E_{m, \sigma}.$
It is easy to see that  (\ref{eq:coupcoll}) provides the claimed bound of Lemma~\ref{lemma:avg_jump}.

\medskip

\noindent
{\em Case 2.} $q = (q_1, \dots, q_{\sigma}) \neq (\frac{m}{\sigma}, \dots, \frac{m}{\sigma}).$
Assume, w.l.o.g., that  $q_{1} \geq q_{2} \geq \dots \geq q_{\sigma}.$ We shall argue by cases according to the magnitude of 
$q_{1}.$

\smallskip

\noindent
{\em Subcase 2.1.} Suppose $q_{1} = \frac{m}{\sigma} + \Omega\left(\sqrt{\frac{m \ln \sigma}{\sigma}}\right).$ 
Let us consider again the analogy with the
coupon collector who has to collect $q_i$ copies of coupons of type $i,$ with $i=1, \dots, \sigma.$
Clearly the collection is not completed until the $q_{1}$'th copy of the coupon of type $1$ has been collected. 
We can model the collection of these type-1 coupons 
as a sequence of Bernoulli trials with probability of success $1/\sigma.$  
The expected waiting time until the $q_1$'th success is $\sigma q_1$ and 
from the previous observation this is also a lower bound on $E_{m, \sigma}.$ Thus,   
$$E_{m,\sigma} \geq \sigma q_1 = \sigma \left( \frac{m}{\sigma} + 
\Omega \left(\sqrt{\frac{m \ln \sigma}{\sigma}}\right) \right) = 
\Omega\left(m+ \sqrt{m \sigma \ln \sigma}\right),$$
which confirms the bound claimed, also in this case.

\smallskip

\noindent
{\em Subcase 2.2.} Finally, assume that $q_{1} = \frac{m}{\sigma} + o\left(\sqrt{\frac{m \ln \sigma}{\sigma}}\right).$
Then, for the smallest component $q_{\sigma}$ of $q$ we have 
$q_{\sigma} \geq m - (\sigma - 1) q_{1} = \frac{m}{\sigma} - o\left(\sqrt{m \sigma\ln \sigma}\right).$
Consider now the balanced Parikh vector $q' = (q_{\sigma}, \dots, q_{\sigma}).$ We have that $q'\leq q$ and 
 $|q'| = \sigma q_{\sigma}.$ By the analysis of Case 1., above, on balanced Parikh vectors, and observing that
collecting $q$ implies collecting $q'$ also, it follows that

\begin{eqnarray*}
E_{m, \sigma} & \geq & E_{\sigma q_{\sigma}, \sigma} \\
&=& \Omega\left( \sigma q_{\sigma} + \sqrt{\sigma^2 q_{\sigma} \ln \sigma} \right)\\
&=& \Omega\left( 
\sigma \left(\frac{m}{\sigma} - o\left(\sqrt{m\sigma \ln \sigma} \right)\right)+
\sqrt{\sigma^2 \left(
\frac{m}{\sigma} - o\left(\sqrt{m \sigma\ln \sigma}\right)
\right)
\ln \sigma}
\right) \\
&=& \Omega\left( 
m -  o\left(\sigma\sqrt{m \sigma \ln \sigma}\right) +
\sqrt{m\sigma \ln \sigma 
- o\left(\sigma^2 \ln \sigma \sqrt{m \sigma\ln \sigma}\right)
}
\right),
\end{eqnarray*}
in agreement with the bound claimed. This completes the proof.

\subsection{Simulations}

We implemented the Jumping Algorithm in C++ in order to study the number of jumps $J$. We ran it on random strings of
different lengths and over different alphabet sizes. The underlying probability model is an i.i.d. model with uniform
distribution. We sampled random query vectors with length between $\log n$ ($=\log_2 n$) and $\sqrt{n}$, where $n$ is the
length of the string. Our queries were of one of two types:

\begin{enumerate}
\item Quasi-balanced Parikh vectors: Of the form $(q_1,\ldots,q_{\sigma})$ with $q_i\in (x-\epsilon, x+\epsilon)$, and
$x$ running from $\log n / \sigma$ to $\sqrt{n}/\sigma$. For simplicity, we fixed $\epsilon=10$ in all our
experiments, and sampled uniformly at random from all quasi-balanced vectors around each $x$.
\item Random Parikh vectors with fixed length $m$. These were sampled uniformly at random from the space of all
Parikh vectors with length $m$.
\end{enumerate}

The rationale for using quasi-balanced queries is that those are clearly worst-case for the number of jumps $J$, since
$J$ depends on the shift length, which in turn depends on $\textsc{firstfit}(prv(L)+q)$. Since we are searching in a random string
with uniform character distribution, we can expect to have minimal $\textsc{firstfit}(prv(L)+q)$ if $q$ is close to 
balanced, i.e.\ if all entries $q_i$ are roughly the same. This is confirmed by our experimental results which show
that $J$ decreases dramatically if the queries are not balanced (Fig.~\ref{fig:simulations2}, right).

We ran experiments on random strings over different alphabet sizes, and observe that our average case analysis agrees
well with the simulation results for random strings and random quasi-balanced query
vectors. Plots for $n=10^5$ and $n=10^6$ with alphabet sizes $\sigma = 2,4,16$ resp.\
$\sigma = 4,16$ are shown in Fig.~\ref{fig:simulations1}.

In Fig.~\ref{fig:scott}  we show comparisons between the running time of the Jumping algorithm and that of the 
simple window algorithm.  
The simulations over random strings and Parikh vectors of different sizes appear to perfectly 
agree with the guarantees provided by our asymptotic analyses. This is of particular importance from the 
point of view of the applications, as it shows that the complexity analysis does not hide big constants. 

\begin{figure}
\begin{center}
\begin{minipage}{6cm}
\includegraphics[scale=0.5]{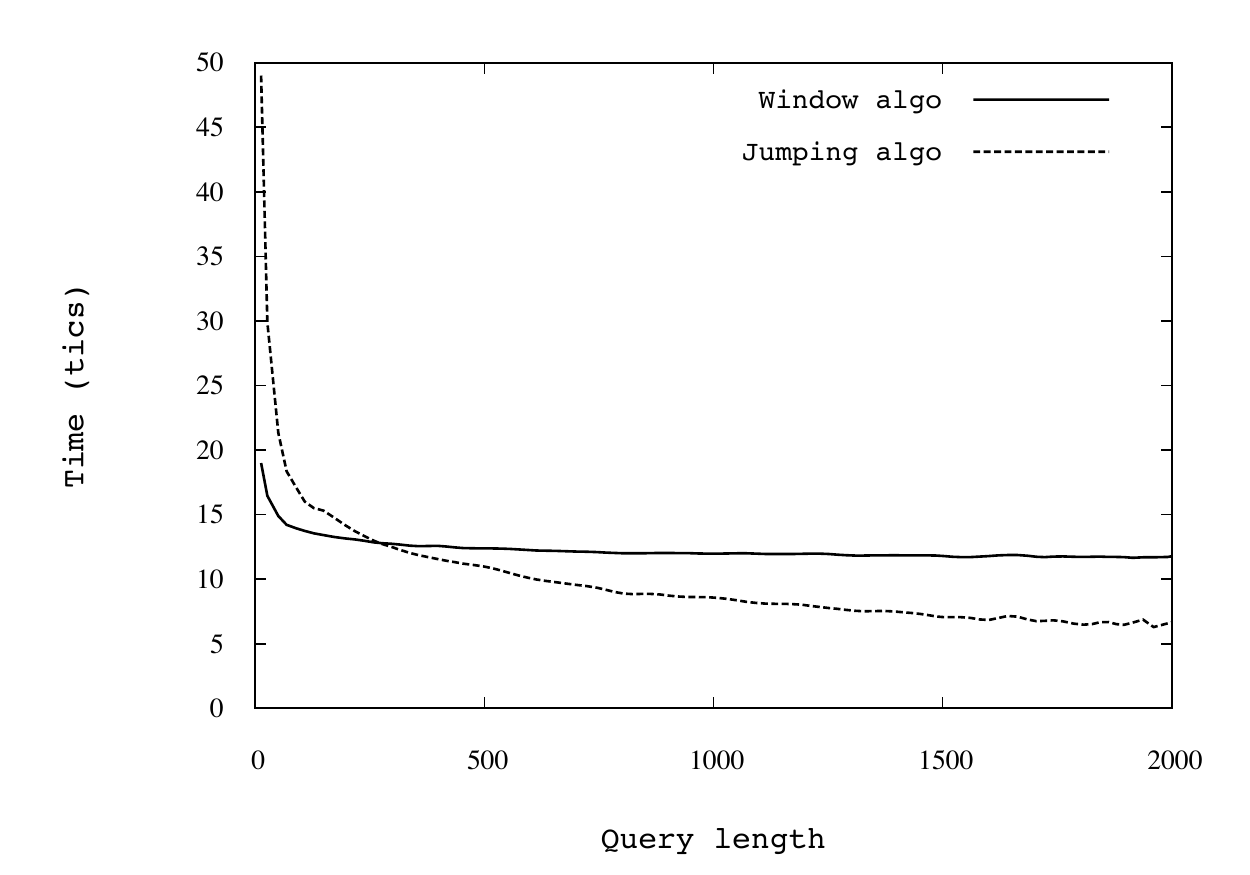}
\end{minipage}
\begin{minipage}{6cm}
\includegraphics[scale=0.5]{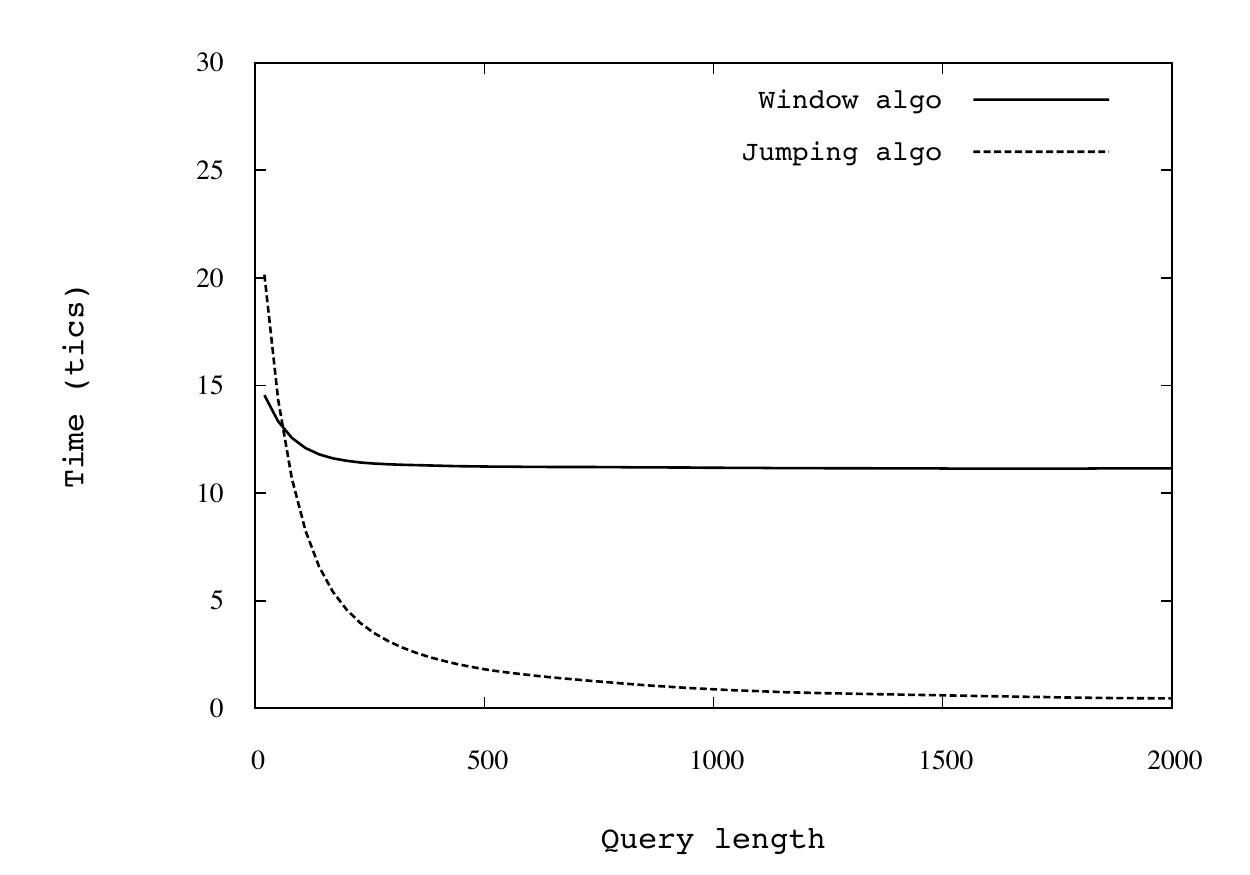}
\end{minipage}
\end{center}
\caption{\label{fig:scott}Running time comparisons between the Jumping Algorithm and the window algorithm. 
The text is a random string (uniform i.i.d.) of size $9\:000\:000$ from a four letter alphabet. Parikh vectors of different sizes between 10 and 2000 were 
randomly generated and the results averaged over all queries of the same size. On the left are the results for quasi-balanced Parikh vectors (cf.~text). 
On the right  are the results for random Parikh vectors.}
\end{figure}

To see how our algorithm behaves on non-random strings, we downloaded human DNA sequences from GenBank~\cite{genbank}
and ran the Jumping Algorithm with random quasi-balanced queries on them.
We found that the algorithm performs 2 to 10 times fewer jumps on these DNA strings than on random strings of the same
length, with the gain increasing as $n$ increases. We show the results on a DNA sequence of $1$ million bp (from
Chromosome 11) in comparison with the average over 10 random strings of the same length (Fig.~\ref{fig:simulations2},
left).

\begin{figure}
\begin{center}
\begin{minipage}{6cm}
\includegraphics[scale=0.5]{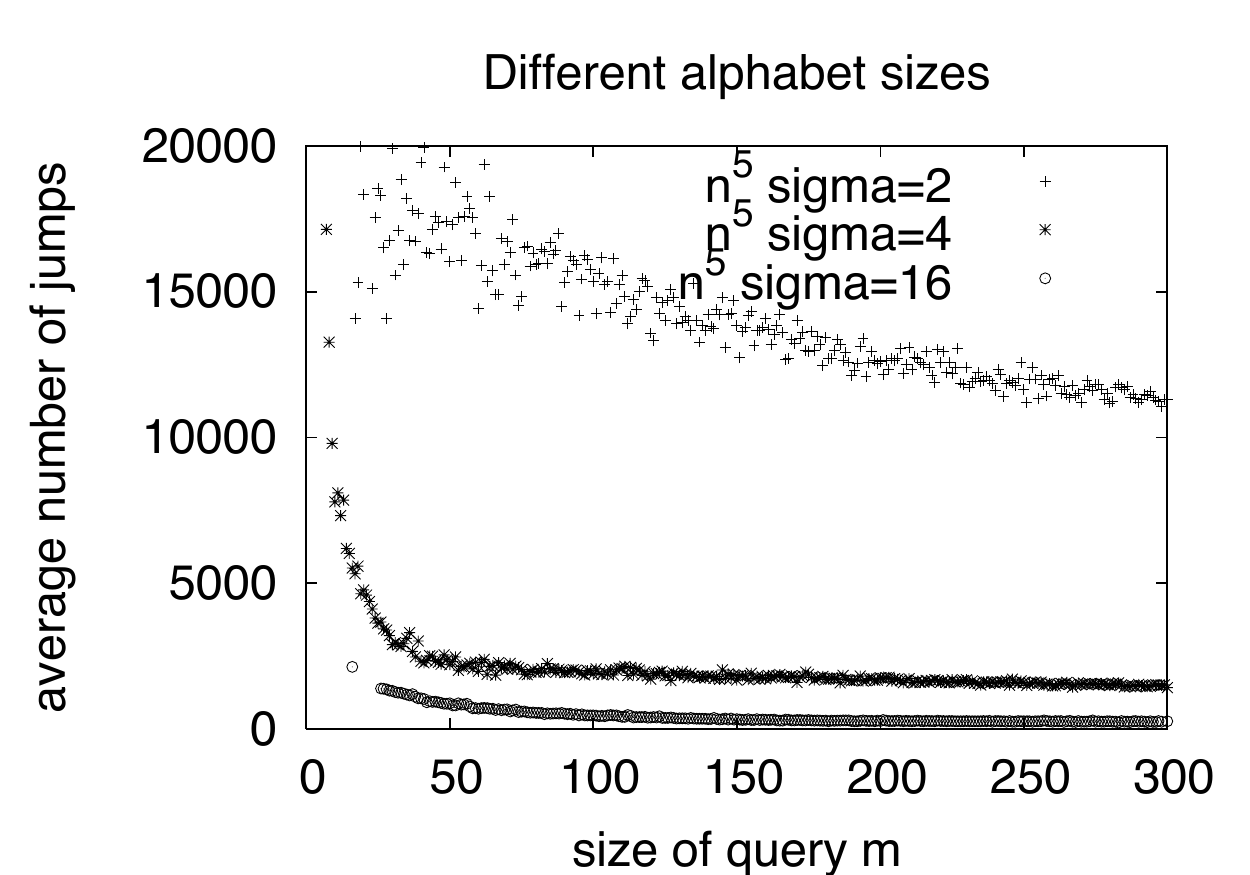}
\end{minipage}
\begin{minipage}{6cm}
\includegraphics[scale=0.5]{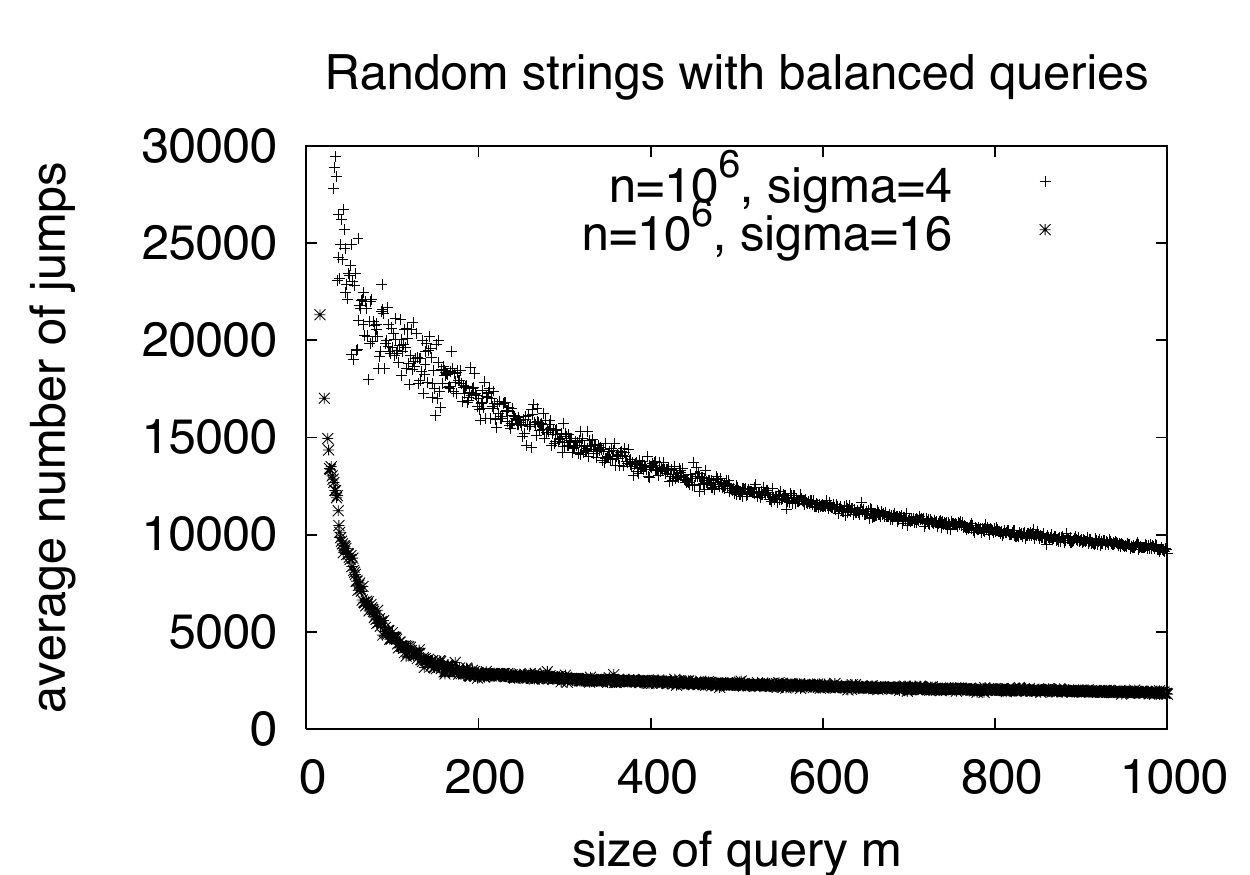}
\end{minipage}
\end{center}
\caption{\label{fig:simulations1}Number of jumps for different alphabet sizes for random strings of size $100\:000$
(left) and $1\: 000\: 000$ (right). All queries are randomly generated quasi-balanced Parikh vectors (cf.\ text). Data
averaged over 10 strings and all random queries of same length.}
\end{figure}

\begin{figure}
\begin{center}
\begin{minipage}{6cm}
\includegraphics[scale=0.5]{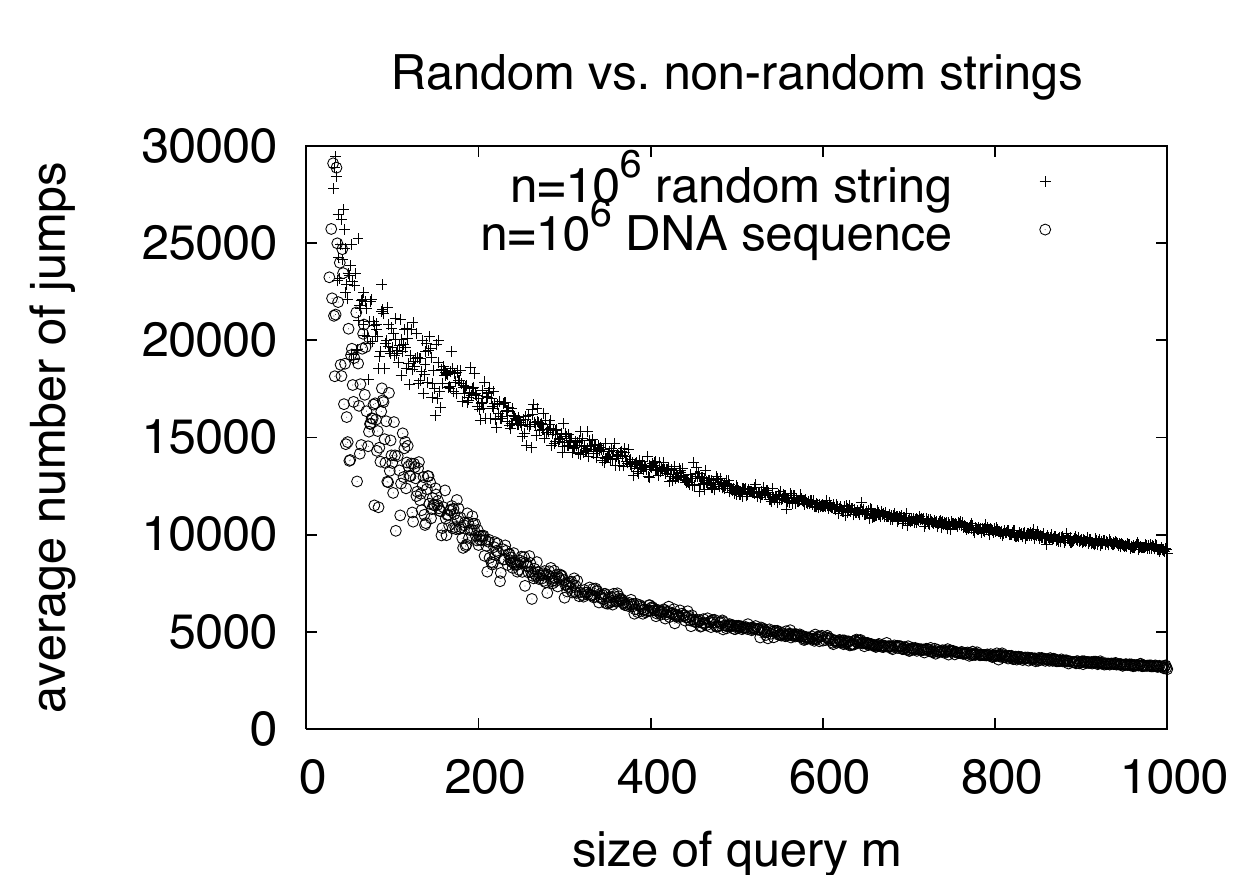}
\end{minipage}
\begin{minipage}{6cm}
\includegraphics[scale=0.5]{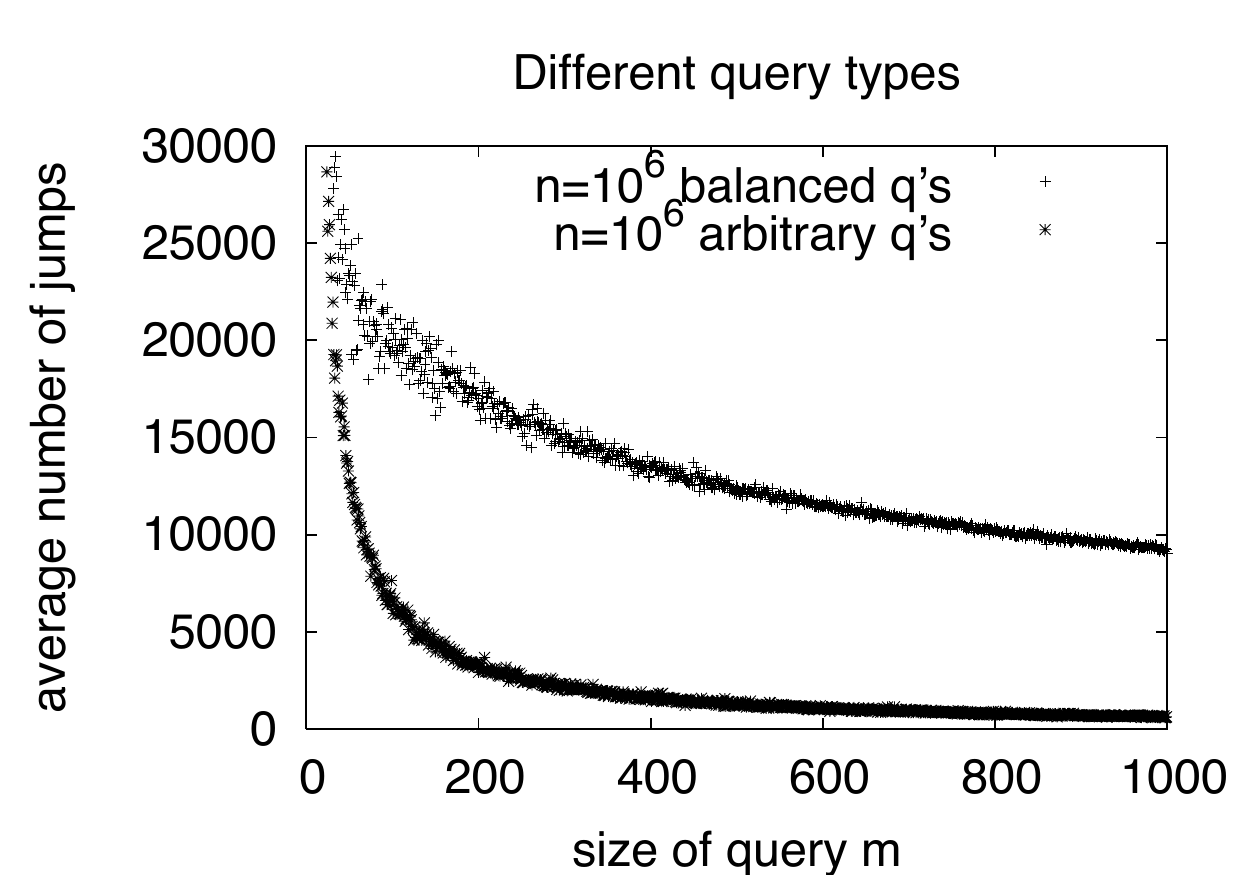}
\end{minipage}
\end{center}
\caption{\label{fig:simulations2}Number of jumps in random vs.\ nonrandom strings: Random strings over an alphabet
of size $4$ vs.\ a DNA sequence, all of length $1\: 000 \: 0000$, random quasi-balanced query vectors. Data
averaged over 10 random strings and
all queries with the same length (left). Comparison of quasi-balanced vs.\ arbitrary query vectors over random strings, 
alphabet size $4$, length $1\:000\:000$, 10 strings. The data shown are averaged over all queries with same length
$m$ (right).}
\end{figure}

\section{Conclusion}

Our simulations appear to confirm that in practice the performance of the Jumping Algorithm is well predicted by the
average case analysis we proposed. A more precise 
analysis is needed, however. Our approach seems unlikely to lead to any refined average case analysis since that would imply
improved results for the intricate variant of the coupon collector problem of~\cite{May08}.

Moreover, in order to better simulate DNA or other biological data, random string models other
than uniform i.i.d.\ should also be analysed, such as first or higher order Markov chains.

We remark that our wavelet tree variant of the Jumping Algorithm, which uses rank/select operations only,
opens a new perspective on the study of Parikh vector matching. We have made another family of approximate pattern matching
problems accessible to the use of self-indexing data structures~\cite{NavMaek07}. We are particularly interested in
compressed data structures which allow
fast execution of rank and select operations, while at the same time using reduced storage space for the text. Thus, 
every step forward in this very active area can provide improvements for our problem.

\subsection*{Acknowledgements}
We thank Gonzalo Navarro for fruitful discussions.

\bibliographystyle{abbrv}

\end{document}